%% file: pas.tex
\documentclass[onecolumn]{IEEEtran} 
\pdfoutput=1
\input{preamble}

\usepackage{cite}
\usepackage{printlen}
\usepackage{siunitx}
\usepackage{csquotes}
\usepackage{caption}
\usepackage{amsmath}
\usepackage{standalone}

\title{Information Rates and Error Exponents for Probabilistic Amplitude Shaping}

\author{%
	\IEEEauthorblockN{Rana Ali Amjad\\}
	\IEEEauthorblockA{Institute of Communications Engineering \\
		Technical University of Munich\\
		ranaali.amjad@tum.de}
}
\newcommand{\mc}[1]{\mathcal{#1}}
\newcommand{\ent}[1]{H\left( #1\right)}
\newcommand{\rent}[2]{H_{#1}\left( #2\right)}
\newcommand{\mi}[2]{I\left( #1;#2\right)}
\newcommand{\kl}[2]{D\left( #1 \| #2\right)}
\newcommand{\Pb}[1]{\Pr\left[#1\right]}
\newcommand{\pnd}[1]{{P}_{#1}^{n}}
\newcommand{\typ}[1]{\mathcal{T}_n(#1)}
\newtheorem{definition}{Definition}
\newcommand{\exop}{\mathbb{E}}

\DeclareMathOperator{\supp}{supp}

\graphicspath{{figures/}}

\IEEEoverridecommandlockouts

\pgfplotsset{
width=0.7\textwidth,
height=0.42\textheight
}
\newtheorem{remark}{Remark}

\begin{document}
\maketitle

\begin{abstract}
Probabilistic Amplitude Shaping (PAS) is a coded-modulation scheme in which the encoder is a concatenation of a distribution matcher with a systematic Forward Error Correction (FEC) code. For reduced computational complexity the decoder can be chosen as a concatenation of a mismatched FEC decoder and dematcher. This work studies the theoretic limits of PAS. The classical joint source-channel coding (JSCC) setup is modified to include systematic FEC and the mismatched FEC decoder. At each step error exponents and achievable rates for the corresponding setup are derived.   
\end{abstract}

\section{Introduction}
A code can approach rates close to the channel capacity of a noisy Discrete Memoryless Channel (DMC) only if the empirical distributions of (not too large) subblocks of the code are \enquote{close} to the Cartesian products of the capacity achieving distribution $P_{X^*}$ \cite{shamaiempirical}, e.g., the Additive White Gaussian Noise (AWGN) channel with uniformly distributed equidistant input signal points has a gap of up to 1.53 dB from the channel capacity. Various methods have been proposed to close this shaping gap, see \cite{georgPAS} for a literature review on this topic. One such method is Probabilistic Amplitude Shaping (PAS) \cite{georgPAS}. PAS has been applied in various communication scenarios including \cite{tobiasshaping} and \cite{fabianshaping} showing significant rate gains as compared to other coded modulation techniques.  Furthermore, in \cite{georgPAS}, the authors show that PAS provides a flexible and low complexity mechanism to adapt rates to changing channel conditions. The method has been implemented in submarine optical fiber with record data transmission rates \cite{fieldtrialnokia1,fieldtrialnokia2} and in a German nationwide fiber optic ring \cite{fieldtrialgerman}. It has been proposed for Digital Subscriber Line (DSL) standards \cite{mgfast}.

In this work we study PAS information rates and error exponents. We adapt the classical Joint Source-Channel Coding (JSCC) setup to include salient features of PAS, i.e., the systematic encoding of a non uniform source and the mismatched decoding. This approach differs from the one taken in \cite{bochererrates} where the focus is solely on mismatched decoding over a codebook larger than the set of transmitted codewords.

The paper is organized as follows.  In Sec.~\ref{sec:prelim} we review the key aspects of Gallager's proof of the coding theorem for a non-uniform source \cite{gallagerexponent, gallager1968information}. In Sec. \ref{sec:pas} we introduce PAS  and discuss how one can analyze it using the JSCC framework. In Sec.~\ref{sec:sys} and Sec.~\ref{sec:mis} we consider two JSCC scenarios, each having one modification as compared to the setup in Sec.~\ref{sec:prelim}. Sec.~\ref{sec:sys} deals with a systematic encoder for transmitting messages from a non uniform Discrete Memoryless Source (DMS) over a noisy channel. Sec.~\ref{sec:mis} deals with a mismatched Maximum Aposteriori Probability (MAP) Forward Error Correction (FEC) decoder. Unlike the previous works on mismatched decoding (e.g., \cite{scarlettmismatch, csiszar1995channel}) where the authors dealt with channel mismatch and/or complexity constraints on the decoder, we look at the mismatch of source statistics and associated complexity at the decoder. Sec.~\ref{sec:paserror} discusses the final setup corresponding to PAS.
\section{Preliminaries}\label{sec:prelim}
 Let $P_{Z}$ be a $n$-type probability distribution \cite{geigeroptimal} over  some finite alphabet $\mc{Z}$ for some positive integer $n$. The set of all sequences $z^n \in \mc{Z}^n$ having empirical distribution $P_{Z}$ is known as a type set and denoted by $\typ{P_{Z}}$. The cardinality $\left|\typ{P_{Z}}\right|$ of $\typ{P_{Z}}$ is bounded as \cite[Th.~11.1.3]{cover2006elements}
\begin{align}
\frac{1}{(n+1)^{\left|\mc{Z}\right|}} 2^{n\ent{P_{Z}}}\ &\leq \left|\typ{P_{Z}}\right| \leq 2^{n\ent{P_{Z}}} \label{eq:typsize}
\end{align}
where $\ent{\cdot}$ denotes Shannon entropy. 
Let $P_{\bar{Z}}$ denote another probability distribution over $\mc{Z}$. Then for any $z^n \in \typ{P_{Z}}$, we have \cite[Th.~11.1.2]{cover2006elements}
\begin{align}
P_{\bar{Z}}^n(z^n) = 2^{-n\left(\ent{P_{Z}} + \kl{P_{Z}}{P_{\bar{Z}}}\right)} \label{eq:typprob}
\end{align}
\subsection{Classical Joint-Source Channel Coding Setup}
\begin{figure}
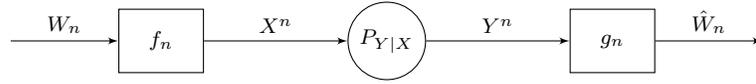

	\centering
	\footnotesize
    \includestandalone{figures/classic/classic}
	\caption{Classical JSCC setup}
	\label{fig:classic}
\end{figure}
Fig.~\ref{fig:classic} shows the classic JSCC setup. The DMC is denoted by $P_{Y|X}$, where $\mc{X}$ and $\mc{Y}$ are finite sets representing the input and output alphabet of the DMC respectively. Although we study DMCs in this paper, the results extend to continuous output alphabets in a straightforward manner. The source message $W_n$ takes values in a finite set $\mc{W}_n$ according to the probability distribution $Q_n$. The JSCC encoder is denoted by $f_n:\mc{W}_n\to\mc{X}^n$ and the JSCC decoder is denoted by $g_n:\mc{Y}^n\to\mc{W}_n$, where $n$ is a positive integer denoting the number of channel uses for transmitting the message $W_n$. The rate of the code is  
\begin{align}
R_n &= \frac{\ent{Q_n}}{n}
\end{align}
We focus on $\ent{Q_n}$ instead of $\log |\mc{W}_n|$ since we allow $W_n$ to be non-uniformly distributed. The average block error probability $\Pb{\hat{W}_n \neq W_n}$ is denoted by $P_{e,n}$. The optimal decoder in the sense of minimizing $P_{e,n}$ for a given code/encoder is the MAP decoder \cite{gallager1968information}. 
\begin{remark}
	$\mathcal{W}_n$ and $Q_n$, and there by their cardinalities,  change with blocklength $n$. This is obvious  in Sec.~\ref{sec:sys} and Sec.~\ref{sec:paserror} when the encoders are systematic but implicit in Sec.~\ref{subsec:gallager} and Sec.~\ref{sec:mis} where the encoders are non-systematic.  
\end{remark}
\subsection{Channel Capacity and Error Exponent}\label{subsec:gallager}
\begin{definition}[Achievable Rate]
	A rate $R$ is achievable if there exists a sequence of encoders $f_n$ and corresponding decoders $g_n$ such that  
	\begin{align}
	P_{e,n} &\overset{n\to\infty}{\longrightarrow} 0 \\
	R_n&\overset{n\to\infty}{\longrightarrow} R
	\end{align}
\end{definition}	
Gallager discussed achievable rates and error exponents for transmitting a non-uniform source message for the setup in Fig. \ref{fig:classic}. 
\begin{theorem}[\protect{\cite[Prob. 5.16]{gallager1968information}}]
	\label{th:gallagernu}
	For the setup in Fig.~\ref{fig:classic}, any channel input distribution $P_X$ and any $n$, there exists an encoder $f_n$ and a corresponding MAP decoder $g_n^{\scriptscriptstyle \text{MAP}}$ s.t.  
	\begin{align}
	P_{e,n} \leq 2^{-nE_G} \label{eq:expnu} 
	\end{align}
	where the error exponent $E_G$ is  
	\begin{align}\label{eq:gallagernonexp}
	E_G = \max_{0 \leq \rho \leq 1} \left[E_0 - \frac{\rho}{n} \rent{\frac{1}{1+\rho}}{Q_n}\right]
	\end{align}	
	where $\rent{\alpha}{\cdot}$ is the Renyi entropy of order $\alpha$ and $E_0$ is the shorthand notation for  
	\begin{align}
	E_0(\rho,P_X) = -\log\sum\limits_{y}\left\{\sum\limits_{x}P_X(x)P_{Y|X}(y|x)^{\frac{1}{1+\rho}}\right\}^{1+\rho}.
	\end{align}
\end{theorem} 
In order to determine the achievable rates we calculate the maximum $R$ for which $E_G$ is positive. For this purpose the interesting region to study is around $\rho=0$. We compute  
\begin{align}
\left. \rho \rent{\frac{1}{1+\rho}}{Q_n} \right|_{\rho=0} =0,& \qquad   E_0(0,P_X) =0 \\ 
\left. \frac{d}{d\rho}\rho \rent{\frac{1}{1+\rho}}{Q_n} \right|_{\rho=0} &= \ent{Q_n} =nR_n\\
\left. \frac{dE_0}{d\rho} \right|_{\rho=0} &= \mi{X}{Y}
\end{align}
where  $\mi{\cdot}{\cdot}$ denotes mutual information. We conclude that, as long as $ R_n < \mi{X}{Y}$, $E_G$ is positive and hence $P_{e,n}$ decays to $0$ exponentially fast in $n$.  Hence $\mi{X}{Y}$ is an achievable rate for a fixed $P_X$. Optimizing over $P_X$ we have  
\begin{align}\label{eq:capnu}
C = \max_{P_X} \mi{X}{Y} 
\end{align}
which is the capacity of a DMC.  In the following, we will follow a similar approach to derive error exponents by upper bounding $P_{e,n}$ and achievable rates by optimizing error exponents.
\section{PAS as Joint Source-Channel Coding}\label{sec:pas}
\begin{figure}
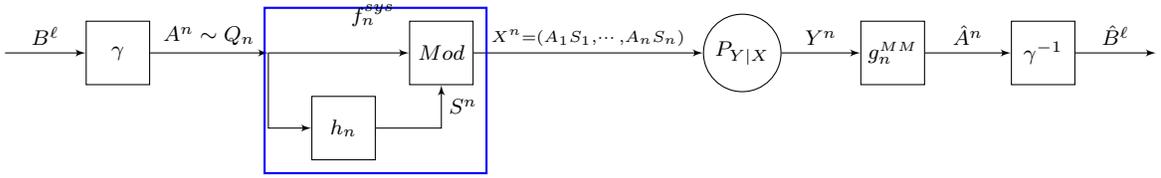

	\centering
	\footnotesize
	\includestandalone{figures/PASsetup/setup}
	\caption{Communication using PAS}
	\label{fig:passetup}
\end{figure}
A maximizer in \eqref{eq:capnu} is a capacity achieving distribution and is denoted by $P_{X^*}$. For unique $P_{X^*}$, a sequence of codes can achieve capacity only if the $k^{th}$ empirical distribution of the codebook (under certain regularity conditions) is \enquote{close} to the product distribution $P_{X^*}^k$ as long as the $k$ is not too large and the blocklength $n$ becomes sufficiently large \cite{shamaiempirical}. For many channels of practical interest such as the AWGN channel with average input power constraint, $P_{X^*}$ is non-uniform. To approach the capacity for such channels, one needs to shape the signal so that the channel input mimics $P_{X^*}$. Various methods have been proposed for signal shaping, including many-to-one mapping \cite{gallager1968information} and trellis shaping \cite{forney1992trellis}.

PAS is a coded modulation scheme which decouples the task of signal shaping from FEC \cite{georgPAS}. This decoupling allows for an efficient implementation of the scheme using off-the-shelf systematic FEC codes. In addition to closing the shaping gap, PAS provides a flexible rate adaption mechanism to adapt to the changing channel conditions. 

Communication using PAS is shown in Fig.~\ref{fig:passetup}. The focus in PAS is on channels that can be represented as $\mc{X} = \mc{A}\times\mc{S}$ for some finite sets $\mc{A}$ and $\mc{S}$. 
\begin{remark}
	In the context of PAS (and for coded-modulation in general) $\mc{X} = \mc{A}\times\mc{S}$ denotes set partitioning, i.e., $\mc{X}$ is partitioned into $|\mc{S}|$ sets, each of size $|\mc{A}|$.  We define two RVs $A$ and $S$ such that $S$ determines the partition which $X$ belongs to and $A$ represents the value that $X$ takes inside this partition. This is represented as $X=AS$. For example if $\mc{X}$ represents ASK modulation, then one possible partition is where $\mc{S}$ determines the sign and $\mc{A}$ determines the amplitude.
\end{remark}

PAS works with channel input distributions of the form $P_X = P_AP_S$ for for any probability distributions $P_A$ and $P_S$ over $\mc{A}$ and $\mc{S}$ . In practice normally $P_S$ is usually considered to be uniform because linear FEC codes generate parity bits by modulo-$2$ sum operations of multiple source bits which leads to more uniform marginal distributions of the parity bits but in our work we do not impose such restrictions. For many practical channels the restriction to product distribution does not incur any penalty. One such example is an AWGN channel with $2^m$-ASK input constellation. The capacity achieving distribution in this case is symmetric around the origin and hence of the form $P_AP_S$ for a uniform $P_S$ over $\mc{S}=\{0,1\}$. By using a Boltzmann distribution (which is of the form $P_AP_S$) one can effectively close the shaping gap for a wide range of SNR values \cite{hubershaping, georgPAS}. In \cite{georgPAS,tobiasshaping} and  \cite{gianluigishaping} the authors discuss the performance gains of using such a shaped distribution to communicate over different channels with ASK and QAM modulations.

In Fig.~\ref{fig:passetup}, $\gamma$ and $\gamma^{-1}$ represent the distribution matcher and dematcher respectively. The aim of $\gamma$ is to invertibly transform the source message $B^\ell$ to look as if it was generated by a Discrete Memoryless Source (DMS) $P_A$. The fundamental limits of distribution matching have been discussed in \cite{bocherer2014informational}. In \cite{schulte2016constant} the authors proposed CCDM, a practical distribution matcher that produces output sequences from a chosen type set. For the purpose of theoretical analysis, we will abstract the concept of a distribution matcher by assuming that the input $A^n$ to $f_n^{sys}$ is distributed according to $Q_n$ and we adapt $Q_n$ to fit to what one expects from the output statistics of a distribution matcher. In subsequent sections we will hence assume that the source alphabet $\mc{W}_n$ is $\mc{A}^n$ in the context of JSCC. Since distribution matching is invertible, using $\ent{Q_n}$ in the rate expression is the right metric because this is the entropy of the source message at the input of the distribution matcher. 

$f_n^{sys}$ transforms $A^n =A_1\cdots A_n$ to the channel input $X^n= (A_1S_1,\cdots, A_nS_n)$ where $S^n$ are the parity symbols generated by the FEC code $h_n$. For convenience we abuse the notation to define $(A^n, S^n) = (A_1S_1,\cdots, A_nS_n)$, hence $X^n = (A^n, S^n)$. $f_n^{sys}$ can be thought of as a systematic channel code since $A^n$ is passed from the input to the output unchanged. The focus on such systematic codes to generate redundancy in PAS is to feed the shaped output $A^n$ of the distribution matcher to the channel unaltered in order to close the shaping gap.

Although the statistics of $A^n$ do not exactly match the statistics of a DMS in general, at the decoder  $g_n^{\scriptscriptstyle MM}$ of a PAS system it is normally assumed that $A^n$ is generated by a DMS $P_A$. This leads to a significant reduction in the decoding complexity. To analyze this we use the framework of mismatched decoding where we restrict ourselves to using a mismatched MAP decoder $g_n^{\scriptscriptstyle MMAP}$ instead of the true MAP decoder $g_n^{\scriptscriptstyle \text{MAP}}$.

Although the problem of coded modulation has been traditionally dealt in the framework of channel coding, the view of various aspects of PAS taken in this section, namely the use of \emph{systematic} FEC codes and the abstraction of the distribution matching process as a non uniform source over $\mc{A}^n$ makes JSCC the suitable framework to analyze PAS. Fig.~\ref{fig:pastheory} presents the JSCC setup that takes into account various aspects of PAS communication that are important for analysis. Our goal in this work is to analyze this setup to understand the theoretical capabilities of PAS.
 \begin{figure}
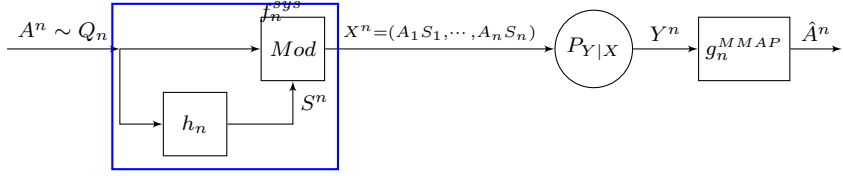

	\centering
	\footnotesize
     \includestandalone{figures/PAStheory/setup}
	\caption{JSCC setup for PAS}
	\label{fig:pastheory}
\end{figure}
\section{Systematic Encoding: Coding Theorem and Error Exponent}\label{sec:sys}
Consider the setup in Fig.~\ref{fig:syssetup} where $Q_n = P_A^n$, $\mc{X} = \mc{A}\times\mc{S}$ and we will use the MAP decoder for the analysis:
\begin{align}
	\hat{a}^n = g_n^{\scriptscriptstyle \text{MAP}}(y^n) &= \argmax\limits_{a^n \in \mc{A}^n} \pnd{Y|X}(y^n|f_n^{sys}(a^n)) P_A^n(a^n)
\end{align}
\begin{theorem}\label{th:sys}
	For  $Q_n=P_A^n$ and any $P_S$, there exists a systematic encoder $f_n^{sys}$ which when used with $g_n^{\scriptscriptstyle \text{MAP}}$ has $P_{e,n}$ upper bounded as   
	\begin{align}
		P_{e,n} \leq 2^{-nE_S}
	\end{align}
	where  
	\begin{align}
	 E_S = \max_{0 \leq \rho \leq 1} \log\left[-\sum\limits_{y} \left\{\sum\limits_{a,s} P_S(s) \left\{P_A(a)P(y|(a,s)) \right\}^{\frac{1}{1+\rho}} \right\}^{1+\rho}\right]
	\end{align}
\end{theorem}
\begin{proof}
	$\\$ \textbf{Code Construction:}
	Denote by $\mc{H}_n$ the set of all mappings from $\mc{A}^n$ to $\mc{S}^n$.  We will prove the theorem for a random ensemble of such mappings denoted by $H_n$ such that
	\begin{align}
		\Pb{H_n(a^n) = s^n} = P_S^n(s^n)    \qquad \forall a^n \in  \mc{A}^n, s^n \in \mc{S}^n
	\end{align} 
	i.e., the probability of choosing such a mapping $H_n$ which maps $a^n$ to $s^n$ for any $a^n$ and $s^n$ is $P_S^n(s^n)$. Hence the parity bits for any input $a^n$ are generated randomly according to the distribution $P_S^n$ in the ensemble. Based on $H_n$ we define the random ensemble of systematic encoders $F_n^{sys}$ which outputs $X^n = (a^n, H_n(a^n))$ as the codeword for $a^n$.  
	\\  \textbf{Encoder:} For any $a^n \in \mc{A}^n$, the codeword $x^n$ is defined as follows
	\begin{align}
		x_i = (a_i, s_i)
	\end{align}
	where 
	\begin{align}
		s^n = h_n(a^n)
	\end{align}
	$h_n$ represents a specific instance of the random variable $H_n$. \\
	\textbf{Decoder:}
	We will use MAP decoder. For a given $f_n^{sys}$ we have
	\begin{align}
		\hat{a}^n = g_n^{\scriptscriptstyle \text{MAP}}(y^n) &= \argmax\limits_{a^n \in \mc{A}^n} \pnd{Y|X}(y^n|f_n^{sys}(a^n)) P_A^n(a^n)
	\end{align} 
	\textbf{Analysis:}
	Define
	\begin{align}
		L(\tilde{a^n},a^n) = \frac{P_{Y^n|X^n} (Y^n|F_n^{sys}(\tilde{a}^n)) P_A^n(\tilde{a}^n)}{P_{Y^n|X^n} (Y^n|F_n^{sys}(a^n)) P_A^n(a^n)} \label{eq:likelihood}
	\end{align}
	Given $a^n$ is the message and $Y^n$ is received, the decoder can make an error if, for some $\tilde{a}^n \neq a^n$, we have
	\begin{align}
		L(\tilde{a^n},a^n) \geq 1
	\end{align}
	$P_{e,n}$, when averaged over the random ensemble of codes and the input, can be upperbounded as follows
	\begin{align}
		P_{e,n} &= \sum\limits_{a^n\in \mc{A}^n} P_A^n(a^n)\Pb{g_n^{\scriptscriptstyle \text{MAP}}(Y^n) \neq a^n|A^n=a^n}  \\ 
		&=\sum\limits_{a^n}  P_A^n(a^n) \sum\limits_{x^n \in \mc{X}^n} \Pb{F_n^{sys}(a^n) = x^n}  \sum\limits_{y^n}  P(y^n|x^n) \Pb{g_n^{\scriptscriptstyle \text{MAP}}(y^n) \neq a^n|A^n=a^n, F_n^{sys}(a^n) = x^n} \\
		&\overset{(a)}{=}\sum\limits_{a^n}  P_A^n(a^n) \sum\limits_{s^n \in \mc{S}^n} P_S^n(s^n)  \sum\limits_{y^n}  P(y^n|(a^n,s^n)) \Pb{g_n^{\scriptscriptstyle \text{MAP}}(y^n) \neq a^n|A^n=a^n, F_n^{sys}(a^n) = (a^n,s^n)} \\
		&\overset{(b)}{\leq} \sum\limits_{a^n}  P_A^n(a^n) \sum\limits_{s^n} P_S^n(s^n)  \sum\limits_{y^n}  P(y^n|(a^n,s^n))  \Pb{\left.\left\{\sum\limits_{\tilde{a}^n \neq a^n}L(\tilde{a}^n,a^n)^\eta\right\}^\rho \geq 1 \right| A^n=a^n,F_n^{sys}(a^n)=(a^n,s^n)}\\
		&\overset{(c)}{\leq} \sum\limits_{a^n}  P_A^n(a^n) \sum\limits_{s^n } P_S^n(s^n)  \sum\limits_{y^n}  P(y^n|x^n)  \exop\left[\left.\left\{\sum\limits_{\tilde{a}^n \neq a^n}L(\tilde{a}^n,a^n)^\eta\right\}^\rho \right| A^n=a^n,F_n^{sys}(a^n)=(a^n,s^n)\right] \\
		&\overset{(d)}{\leq} \sum\limits_{a^n}  P_A^n(a^n) \sum\limits_{s^n} P_S^n(s^n)  \sum\limits_{y^n}  P(y^n|(a^n,s^n))  \left[\expop\left\{\left.\sum\limits_{\tilde{a}^n \neq a^n}L(\tilde{a}^n,a^n)^\eta\right| A^n=a^n,F_n^{sys}(a^n)=(a^n,s^n)\right\}^\rho \right] \\
		&\overset{(e)}{=}\sum\limits_{y^n} \left\{\sum\limits_{a^n,s^n} P_S^n(s^n) \left\{P_A^n(a^n)P(y^n|(a^n,s^n)) \right\}^{\frac{1}{1+\rho}} \right\}^{1+\rho} \\
		&\overset{(f)}{=}\left[\sum\limits_{y} \left\{\sum\limits_{a,s} P_S(s) \left\{P_A(a)P(y|(a,s)) \right\}^{\frac{1}{1+\rho}} \right\}^{1+\rho}\right]^n \\
		&=2^{-n\left[-\log\sum\limits_{y} \left\{\sum\limits_{a,s} P_S(s) \left\{P_A(a)P(y|(a,s)) \right\}^{\frac{1}{1+\rho}} \right\}^{1+\rho}\right]} \label{eq:sysfin}
	\end{align}
	where
		\begin{figure}
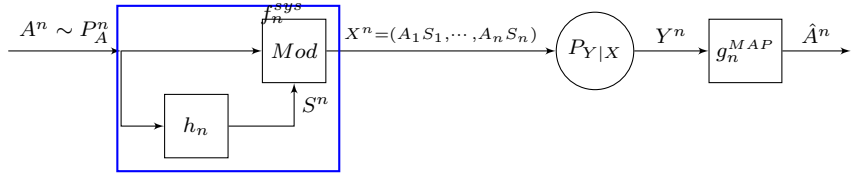

		\centering
		\footnotesize
		\includestandalone{figures/sys/setup}
		\caption{JSCC using systematic encoder}
		\label{fig:syssetup}
	\end{figure}
	\begin{itemize}
		\item (a) follows from the definition of the ensemble $H_n$.
		\item (b) is true for any $\rho \geq 0$ and $\eta \geq 0$.
		\item (c) uses Markov inequality
		\item (d) is a result of using Jensen inequality and restricting $ 0 \leq \rho \leq 1$	
		\item (e) uses $\eta =\frac{1}{1+\rho}$, adds $\tilde{a}^n = a^n$ term in the summation inside expectation  and does a rearrangement of the terms inside the expectation.
		\item (f) is due to the product distributions and the DMC.
	\end{itemize}
	Optimizing \eqref{eq:sysfin} over $\rho$ we get the final exponent $E_S$. Using the standard argument that at least one of the codes in the random ensemble should have $P_{e,n}$ at least as good as the average $P_{e,n}$ of the ensemble, we arrive at the conclusion that there exist at least one systematic encoder which when used with MAP decoder have its $P_{e,n}$ upperbounded as \eqref{eq:sysfin}.
\end{proof}
 Unlike \eqref{eq:gallagernonexp}, we do not have two separate terms in $E_S$, one related to $R_n$ whereas the other dependent on the channel and the channel input distribution. To evaluate the achievable rates we have 
 \begin{align}
 \left.\log\left[\sum\limits_{y} \left\{\sum\limits_{a,s} P_S(s) \left\{P_A(a)P(y|(a,s)) \right\}^{\frac{1}{1+\rho}} \right\}^{1+\rho}\right] \right|_{\rho =0} &= 0 \\
 \left.\frac{d}{d\rho} \left[-\log\sum\limits_{y} \left\{\sum\limits_{a,s} P_S(s) \left\{P_A(a)P(y|(a,s)) \right\}^{\frac{1}{1+\rho}} \right\}^{1+\rho}\right] \right|_{\rho =0} &=  \mi{AS}{Y} - \ent{A} 
 \end{align}
 Hence as long as $\ent{A} < \mi{AS}{Y}$, $E_S$ will be positive. We also know that $R_n = \ent{A}$, hence we conclude that any rate $R< \mi{X}{Y}$ is achievable for $P_X=P_AP_S$.
 
 Note that $\ent{A} < \mi{AS}{Y}$ implies $\ent{S} > \ent{AS|Y} = \ent{X|Y}$. $\ent{S}$ corresponds to the redundancy that we introduce for reliability whereas $ \ent{X|Y}$ is the uncertainty the channel introduces. 
\begin{corollary}
	For uniform $P_S$, we have 
	\begin{align}
       &\sum\limits_{y^n} \left\{\sum\limits_{a^n,s^n} P_S^n(s^n) \left\{P_A^n(a^n)P(y^n|(a^n,s^n)) \right\}^{\frac{1}{1+\rho}} \right\}^{1+\rho} \\
       &\overset{(a)}{=}2^{-n\rho\log|\mc{S}|}\sum\limits_{y^n} \left\{\sum\limits_{a^n, s^n} \left\{P_S^n(s^n) P_A^n(a^n)P(y^n|(a^n,s^n)) \right\}^{\frac{1}{1+\rho}} \right\}^{1+\rho} \\
       &\overset{(b)}{=} 2^{-n\rho\log|\mc{S}|}2^{\rho \rent{\frac{1}{1+\rho}}{X^n|Y^n}} \\	      
       &= 2^{-n\rho \left(\log|\mc{S}| - \rent{\frac{1}{1+\rho}}{X|Y}\right)}    
	\end{align}
	where 
\begin{itemize}
		\item (a) we used the fact that $P_S^n(\cdot) = 2^{-n\log|\mc{S}|}$.
		\item (b) uses the definition of conditional Renyi entropy.
\end{itemize}
     Hence in this case 
		\begin{align}
	       E_S = \max_{0 \leq \rho \leq 1} \left[\rho(\log|\mc{S}|-\rent{\frac{1}{1+\rho}}{X|Y})\right].
      	\end{align}  
      	where the error exponent separates into two terms, one related to the rate whereas the other related to channel and the channel input distribution. 
\end{corollary}

\begin{corollary}\label{rem:affine}
	For $|\mc{X}|=2^m$ we can rename $\mc{X} =\{0,1\}^m$ and define $\mc{S}=\{0,1\}^p$ and $\mc{A}=\{0,1\}^{m-p}$. We can then restrict $\mc{H}_n$ to be the set of all affine mappings from $\mc{A}^n$ to $\mc{S}^n$ and define a uniform random ensemble $H_n$ of such mappings as done by Gallager \cite[Sec 6.2]{gallager1968information}.  This ensemble then can be used to prove Th.~\ref{th:sys}. Hence for such channels any rate $R < \mi{X}{Y}$, where $P_X=P_AP_S$ for uniform $P_S$, is achievable using systematic linear codes. 
\end{corollary}	
\begin{remark}\label{rem:opt}
We have considered the source distribution $P_A$ to be given. In practice we can design the distirbution matcher to (ideally) mimic a DMS our choice, hence $P_A$ then becomes a design choice leading to the following maximum achievable rate expression
\begin{align}
R^* =  \max_{P_A, P_S} \mi{X}{Y}  
\end{align}
where $P_X=P_AP_S$
\end{remark}
\section{Source Statistics Mismatch: Coding Theorem and Error Exponent}\label{sec:mis}
\begin{figure}
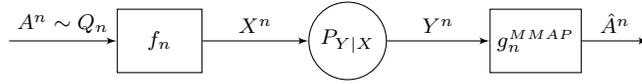

	\centering
	\footnotesize
    \includestandalone{figures/mismatched/mis}
	\caption{Joint-Source Channel Coding with Mismatched MAP decoder}
	\label{fig:mismatch}
\end{figure}
The setup for this section is shown in Fig.~\ref{fig:mismatch}. The only difference to the setup in Fig.~\ref{fig:classic} is that we now use a mismatched decoder, specifically the mismatch is between the actual source statistics $Q_n$ and $P_A^n$ assumed at the decoder. In this section we do not restrict the channel input alphabet to be of the form $\mc{X} = \mc{A}\times\mc{S}$ since we are not focusing on systematic encoders.
\begin{theorem}\label{th:sourcemis}
	Consider the DMC $P_{Y|X}$ with any finite input alphabet $\mc{X}$ and let $P_X$ be any distribution over $\mc{X}$. Let $P_{\bar{A}}$  be some $k$-type distribution and  $P_A$ be any distribution satisfying $P_A \gg P_{\bar{A}}$ $\left(\supp(P_{\bar{A}}) \subseteq \supp(P_A)\right)$, then, for every $n = kj$ where $j$ is a positive integer and  $\supp(Q_n) \subseteq \typ{P_{\bar{A}}}$, there exists an encoder $f_n:\mc{A}^n\to\mc{X}^n$ for this setup, such that when used with the mismatched MAP decoder  
	\begin{align}\label{eq:misdecth}
		\hat{a}^n = \argmax_{a^n \in \mc{A}^n} P_{Y|X}^n(y^n|f_n(a^n)) P_A^n(a^n)
	\end{align}
	has a  $P_{e,n}$ upper bounded by  
	\begin{align}
		P_{e,n} \leq 2^{-nE_M}
	\end{align}
	where  
	\begin{align}\label{eq:expsrc}
		E_M =\max_{0 \leq \rho \leq 1}  \left[E_0 - \frac{\rho}{1+\rho}\left(\kl{P_{\bar{A}}}{P_A}+\ent{P_{\bar{A}}}\right)\right. \\
		\left. - \frac{\rho^2}{1+\rho}\rent{\frac{1}{1+\rho}}{P_A} \right]. \nonumber
	\end{align}
\end{theorem}
\begin{proof} 
     $\\$ \textbf{Code Construction:} Denote by $\mc{F}_n$ the set of all mappings from $\mc{A}^n$ to $\mc{X}^n$.  We will prove the theorem for a random ensemble of such mappings denoted by $F_n$ such that
      \begin{align}
      \Pb{F_n(a^n) = x^n} = P_X^n(x^n)   \qquad \forall a^n, x^n
      \end{align}
      Note that the mappings in $\mc{F}_n$ are defined for all $a^n \in \mc{A}^n$, i.e., every $a^n$ is assigned a codeword, although the source message takes values only in the set $\typ{P_{\bar{A}}}$. This is necessary for the decoding rule in \eqref{eq:misdecth} to be well defined which searches over codewords for all sequences $a^n \in \mc{A}^n$. 
     \\ \textbf{Encoder:} For any $a^n \in \typ{P_{\bar{A}}}$
     \begin{align}
     x^n = f_n(a^n)
     \end{align}
     where $f_n$ is a specific instance of $F_n$.\\
     \textbf{Decoder:} As mentioned in the theorem, for a given $f_n$, we will use the following decoding rule.
     \begin{align}
     \hat{a}^n = g_n^{\scriptscriptstyle \text{MMAP}}(y^n) =  \argmax_{a^n \in \mc{A}^n} P_{Y|X}^n(y^n|f_n(a^n)) P_A^n(a^n) \label{eq:misdecrule}
     \end{align}
     \textbf{Analysis:}
     Define
     \begin{align}
     L(\tilde{a^n},a^n) = \frac{P_{Y^n|X^n} (Y^n|f_n(\tilde{a}^n)) P_A^n(\tilde{a}^n)}{P_{Y^n|X^n} (Y^n|f_n(a^n)) P_A^n(a^n)} \label{eq:likelihood}
     \end{align}
     Given $a^n$ is the message and $Y^n$ is received, the decoder can make an error if, for some $\tilde{a}^n \neq a^n$, we have
     \begin{align}
     L(\tilde{a^n},a^n) \geq 1
     \end{align}
     $P_{e,n}$, when averaged over the random ensemble of codes and the input, can be upperbounded as follows
     \begin{align}
    P_{e,n} &= \sum\limits_{a^n\in \typ{P_{\bar{A}}}} Q(a^n)\Pb{g_n^{\scriptscriptstyle \text{MMAP}}(Y^n) \neq a^n|A^n=a^n}  \\ 
    &=\sum\limits_{a^n\in \typ{P_{\bar{A}}}}  Q(a^n) \sum\limits_{x^n \in \mc{X}^n} \Pb{F_n(a^n) = x^n}  \sum\limits_{y^n}  P(y^n|x^n) \Pb{g_n^{\scriptscriptstyle \text{MMAP}}(y^n) \neq a^n|A^n=a^n, F_n(a^n) = x^n} \\
   &\overset{(a)}{=}\sum\limits_{a^n\in \typ{P_{\bar{A}}}}  Q(a^n) \sum\limits_{x^n} P_X^n(x^n)  \sum\limits_{y^n}  P(y^n|x^n) \Pb{g_n^{\scriptscriptstyle \text{MMAP}}(y^n) \neq a^n|A^n=a^n, F_n(a^n) = x^n} \\    
   &\overset{(b)}{\leq} \sum\limits_{a^n \in \typ{P_{\bar{A}}}}  Q(a^n) \sum\limits_{x^n} P_X^n(x^n)  \sum\limits_{y^n}  P(y^n|x^n) \left[\expop\left\{\left.\sum\limits_{\tilde{a}^n \neq a^n}L(\tilde{a}^n,a^n)^\eta\right| A^n=a^n,F_n(a^n)=x^n\right\}^\rho \right] \\
   &\overset{(c)}{\leq} \sum\limits_{a^n \in\typ{P_{\bar{A}}}} Q(a^n)P_A^n(a^n)^{-\frac{\rho}{1+\rho}}\left\{\sum\limits_{\tilde{a}^n \in \mc{A}^n}P_A^n(\tilde{a}^n)^{\frac{1}{1+\rho}}\right\}^\rho \sum\limits_{y^n}\left\{\sum\limits_{x^n} P(x^n) P(y^n|x^n)^{\frac{1}{1+\rho}}\right\}^{1+\rho} \label{eq:misfin}
        \end{align}
     where
     \begin{itemize}
     	\item (a) follows from the definition of the ensemble $F_n$.
     	\item (b) uses Markov inequality followed by Jensen equality and is true for any $0 \leq \rho \leq 1$ and $\eta \geq 0$.
     	\item (c) includes $\tilde{a}^n = a^n$ in the summation inside expectation, replaces $\eta =\frac{1}{1+\rho}$ and rearranges terms.
     \end{itemize}
 A careful look at the three terms in \eqref{eq:misfin} reveal that
   \begin{align}
   \sum\limits_{a^n \in \typ{P_{\bar{A}}}} Q(a^n)P_A^n(a^n)^{-\frac{\rho}{1+\rho}} \overset{(a)}{=} 2^{\frac{n\rho}{1+\rho}\left(\kl{P_{\bar{A}}}{P_A}+\ent{P_{\bar{A}}}\right)} \sum\limits_{a^n \in \typ{P_{\bar{A}}}} Q(a^n)   &=  2^{\frac{n\rho}{1+\rho}\left(\kl{P_{\bar{A}}}{P_A}+\ent{P_{\bar{A}}}\right)} \\
   \left\{\sum\limits_{\tilde{a}^n \in \mc{A}^n}P_A^n(\tilde{a}^n)^{\frac{1}{1+\rho}}\right\}^\rho &= 2^{\frac{n\rho^2}{1+\rho}H_{\frac{1}{1+\rho}}(P_A)} \\
      \sum\limits_{y^n}\left\{\sum\limits_{x^n} P(x^n) P(y^n|x^n)^{\frac{1}{1+\rho}}\right\}^{1+\rho} &= 2^{-nE_0}
   \end{align}
   where (a) follows from \eqref{eq:typprob}. Hence 
\begin{align}
P_{e,n} \leq 2^{-n\left( E_0 - \frac{\rho}{1+\rho}\left(\kl{P_{\bar{A}}}{P_A}+\ent{P_{\bar{A}}}\right) - \frac{\rho^2}{1+\rho}\rent{\frac{1}{1+\rho}}{P_A} \right)} \label{eq:misfin2}
\end{align}
Finally we can optimize over $\rho$ to obtain $2^{-nE_M}$. Using the standard argument that at least one of the codes in the random ensemble should have $P_{e,n}$ at least as good as the average $P_{e,n}$ of the ensemble, we arrive at the conclusion that there exist at least one such encoder which when used with the mismatched MAP decoder have its $P_{e,n}$ upperbounded as \eqref{eq:misfin2}.
\end{proof}
A few remarks about the theorem and the setup are in order. 
\begin{itemize}
	\item In the theorem, the support of $Q_n$ has been restricted to $\typ{P_{\bar{A}}}$ which can be justified by looking at the operation of practical matchers such as CCDM \cite{patrickccdm}.
	\item  Looking at \eqref{eq:misdecth}, the mismatch is because the decoder assumes that the input is generated by a DMS $P_A$ instead of the actual source distribution $Q_n$. Hence not only does the decoder use the wrong statistics, but it also searches over all sequences in $\mc{A}^n$ instead of only $\typ{P_{\bar{A}}}$ to look for the most probable input sequence. Note that this setup is different from what is studied under the name mismatched decoding in the literature (e.g., \cite{csiszarmismatch, scarlettmismatch}), where the mismatch between the actual channel statistics and the channel statistics assumed at the decoder are discussed. The motivation to analyze a mismatched decoder (in this section and in Sec.~\ref{sec:paserror}) that assumes $A^n$ (the output of the distribution matcher in a PAS system) to be generated by a DMS $P_A$ instead of the true distribution $Q_n$ comes from practical systems to reduce the decoding complexity. If the decoder would consider the true distribution then it will have to not only deal with the FEC constraints introduced in $S^n$ but also with the constraints introduced by the matcher on $A^n$, hence coupling the FEC decoding and dematching process which leads to increased computational complexity.
	\item We need $P_A \gg  P_{\bar{A}}$ because of our choice of the decoder ,i.e., the mismatched MAP decoder.  $P_{Y|X}^n(y^n|f_n(a^n)) P_A^n(a^n)$ will be $0$ for all $a^n \in \typ{P_{\bar{A}}}$ if this condition is not satisfied. 
\end{itemize}
To calculate the achievable rates, note that all 3 terms in the R.H.S of \eqref{eq:expsrc} are $0$ for $\rho =0$ and   
\begin{align}
\scriptstyle \left. \frac{d}{d\rho} \frac{\rho}{1+\rho}\left(\kl{P_{\bar{A}}}{P_A}+\ent{P_{\bar{A}}}\right)\right|_{\rho =0} &= \kl{P_{\bar{A}}}{P_A}+\ent{P_{\bar{A}}} \\
\left.\frac{d}{d\rho} \frac{\rho^2}{1+\rho}\rent{\frac{1}{1+\rho}}{P_A} \right|_{\rho =0} &= 0 \\
\left.\frac{d}{d\rho} E_0 \right|_{\rho =0} &= \mi{X}{Y}
\end{align}
Hence as long as 
\begin{align}
\ent{P_{\bar{A}}} < \mi{X}{Y} - \kl{P_{\bar{A}}}{P_A}
\end{align}
 $E_M$ will be positive.
Since we have $\supp(Q_n)\subseteq\typ{P_{\bar{A}}}$ hence $R_n \leq \frac{1}{n}\log\left|\typ{P_{\bar{A}}}\right|$. Combining this with \eqref{eq:typsize} we get  
\begin{align}
R_n \leq \ent{P_{\bar{A}}}.
\end{align}
By utilizing most of the type set  $\typ{P_{\bar{A}}}$, i.e., $\frac{|\supp(Q_n)|}{|\typ{P_{\bar{A}}}|} \to 1$ and having $Q_n$ \enquote{close} to uniform, the lower bound in \eqref{eq:typsize} leads to the result that any rate below $\mi{X}{Y} - \kl{P_{\bar{A}}}{P_A}$ can be achieved.
$\kl{P_{\bar{A}}}{P_A}$ is the penalty one pays for assuming the wrong input DMS $P_A$ at the decoder.
\begin{corollary}
For $P_A = P_{\bar{A}}$, any $R < \mi{X}{Y}$ is achievable. Note that even for $P_A = P_{\bar{A}}$, it is still a mismatched setup since the encoder input has some arbitrary distribution $Q$ focused on $\typ{P_{\bar{A}}}$ while the decoder assumes a DMS $P_{\bar{A}}$.
\end{corollary}
\begin{remark}
For the scenario in Corollary~\ref{rem:affine} we can restrict $\mc{F}_n$ to be the set of all affine mappings from $\mc{A}^n$ to $\mc{X}^n$ and define a uniform random ensemble $F_n$ of such mappings as done by Gallager \cite[Sec 6.2]{gallager1968information}.  This ensemble then can be used to prove Th.~\ref{th:sourcemis} for uniform $P_X$. Hence any rate $R < \mi{X}{Y} - \kl{P_{\bar{A}}}{P_A}$ for uniform $P_X$ is achievable using linear codes.
\end{remark}
\section{Probabilistic Amplitude Shaping: Achievable Rates and Error Exponent}\label{sec:paserror}
Having looked at systematic encoding and mismatched decoding individually in Sec.~\ref{sec:sys} and Sec.~\ref{sec:mis}, we now combine the two.
 \begin{definition}[Permuter]
	Let $P_Z$ be an $m$-type distribution. A type-$P_Z$ permuter is a function $\phi_{P_Z} : \mc{Z}^n \to \mc{Z}^n$ such that
	\begin{align}
	            \phi_{P_Z}(z^n) &=z^n    \qquad \forall z^n \in \mc{Z}^n \setminus \typ{P_Z} \\
	             \phi_{P_Z}(z^n) &= \phi_{P_Z}'(z^n) \qquad \forall z^n \in  \typ{P_Z}
	\end{align} 
	where $\phi_P'$ is a permutation function from $\typ{P_Z}$ to $\typ{P_Z}$.
\end{definition}
\begin{theorem}\label{th:pasg}
	Let $P_{\bar{A}}$  be some $k$-type distribution. For every $n = kj$ for positive integers $j$ and  $\supp(Q_n) \subseteq \typ{P_{\bar{A}}}$, there exists an encoder $f_n^{sys}(\phi_{P_{\bar{A}}}(\cdot))$, where $f_n^{sys} $ is a systematic encoder and $\phi_{P_{\bar{A}}}$ is a type-$P_{\bar{A}}$ permuter, such that when used with the following mismatched MAP decoder  
	\begin{align}
	\hat{a}^n = g_n^{\scriptscriptstyle \text{MMAP}}(y^n) =  \argmax_{a^n \in \mc{A}^n} P_{Y|X}^n\left(y^n|f_n^{sys}(\phi_{P_{\bar{A}}}(a^n))\right) P_{A}^n(a^n)
	\end{align}
	has a $P_{e,n}$  upper bounded as  
	\begin{align}
	P_{e,n} \leq 2^{-nE_{SM}}
	\end{align}
	where  
	\begin{align}
	E_{SM} = \max_{0 \leq \rho \leq 1} \left(-\log\sum\limits_{y} \left\{\sum\limits_{a,s} P_S(s) \left\{P_A(a)P(y|(a,s)) \right\}^{\frac{1}{1+\rho}} \right\}^{1+\rho} -\alpha(n) -\kl{P_{\bar{A}}}{P_A}\right)
	\end{align}
	where $P_X = P_AP_S$.
	\begin{align}
	\alpha(n) &= |\mc{A}|\frac{\log (n+1)}{n} \overset{n\to \infty}{\longrightarrow} 0
	\end{align}
\end{theorem}  
\begin{proof}
	$\\$\textbf{Code Construction:} Define $H_n$ and $F_n^{sys}$ as in Sec.~\ref{sec:sys}. Besides denote by $\gamma$ the set all possible type-$P_{\bar{A}}$ permuters. Let $\Phi$ be a uniform random variable over $\gamma$.
	
	\textbf{Encoder:}
	    For any $a^n \in \typ{P_{\bar{A}}}$
	    \begin{align}
	    x_i &= (\bar{a}_i, s_i)
	    \end{align}
	    where 
	    \begin{align}
	    s^n = h_n(a^n) & \qquad \bar{a}^n = \phi(a^n) 
	    \end{align}
	    $h_n$ represents any specific instance of the random variable $H_n$ and $\phi$ represents an instance of $\Phi$ \\ 
	\textbf{Decoder:} As mentioned in the theorem, for a given $f_n^{sys}$, we will analyze the following decoding rule. 
	\begin{align}
	\hat{a}^n = g_n^{\scriptscriptstyle \text{MMAP}}(y^n) =  \argmax_{a^n \in \mc{A}^n} P_{Y|X}^n\left(y^n|f_n^{sys}(\phi_{P_{\bar{A}}}(a^n))\right) P_A^n(a^n)
	\end{align}
	              The discussion about mismatched decoding in Sec.~\ref{sec:mis} is also applicable here.	    \\           
	\textbf{Analysis:}
	We know that 
	\begin{align}
	 P_{e,n} &= \sum\limits_{a^n\in \typ{P_{\bar{A}}}} Q(a^n)\Pb{g_n^{\scriptscriptstyle \text{MMAP}}(Y^n) \neq a^n|A^n=a^n}  \\ 
	\end{align}
	Define
	\begin{align}
	L(\tilde{a^n},a^n) = \frac{P_{Y|X}^n\left(y^n|f_n^{sys}(\phi_{P_{\bar{A}}}(\tilde{a}^n))\right) P_A^n(\tilde{a}^n)}{P_{Y|X}^n\left(y^n|f_n^{sys}(\phi_{P_{\bar{A}}}(a^n))\right) P_A^n(a^n)} \label{eq:likelihood}
	\end{align}
	Given $a^n \in \typ{P_{\bar{A}}}$ is the message and $Y^n$ is received, the decoder can make an error if, for some $\tilde{a}^n \neq a^n$, we have
	\begin{align}
	L(\tilde{a^n},a^n) \geq 1
	\end{align}
	When averaging over the random ensemble we get,
	\begin{align}
	&\Pb{g_n^{\scriptscriptstyle \text{MM}}(Y^n) \neq a^n|A^n=a^n}   \overset{(a)}{\leq}  \exop\left[\left.\left\{\sum\limits_{\substack{\tilde{a}^n \in \mc{A}^n \\ \tilde{a}^n \neq a^n}} L(\tilde{a}^n,a^n)^s\right\}^\rho\right|A^n=a^n\right] \\	
	 &=\sum\limits_{\bar{a}^n\in \typ{P_{\bar{A}}}}  \Pb{\bar{a}^n =\Phi(a^n)} \sum\limits_{x^n \in \mc{X}^n} \Pb{F_n^{sys}(\bar{a}^n) = x^n}  \sum\limits_{y^n}  P_{Y|X}^n(y^n|x^n)   \exop\left[\left.\left\{\sum\limits_{\substack{\tilde{a}^n \in \mc{A}^n \\ \tilde{a}^n \neq a^n}} L(\tilde{a}^n,a^n)^\eta\right\}^\rho\right|\substack{A^n=a^n\\ \Phi(a^n) = \bar{a}^n \\ F_n^{sys}(\bar{a}^n) = x^n}\right]\\
	 &\overset{(b)}{=}\sum\limits_{\bar{a}^n\in \typ{P_{\bar{A}}}}  \frac{1}{\left|\typ{P_{\bar{A}}}\right|} \sum\limits_{s^n \in \mc{S}^n} P_S^n(s^n)  \sum\limits_{y^n}  P_{Y|X}^n(y^n|(\bar{a}^n,s^n))   \exop\left[\left.\left\{\sum\limits_{\substack{\tilde{a}^n \in \mc{A}^n \\ \tilde{a}^n \neq a^n}} L(\tilde{a}^n,a^n)^\eta\right\}^\rho\right|\substack{A^n=a^n\\ \Phi(a^n) = \bar{a}^n \\ F_n^{sys}(\bar{a}^n) = (\bar{a}^n,s^n)}\right]\\
	 &\overset{(c)}{\leq}\sum\limits_{\bar{a}^n\in \typ{P_{\bar{A}}}}  \frac{1}{\left|\typ{P_{\bar{A}}}\right|} \sum\limits_{s^n \in \mc{S}^n} P_S^n(s^n)  \sum\limits_{y^n}  P_{Y|X}^n(y^n|(\bar{a}^n,s^n))   \left[\exop\left.\left\{\sum\limits_{\substack{\tilde{a}^n \in \mc{A}^n \\ \tilde{a}^n \neq a^n}} L(\tilde{a}^n,a^n)^\eta\right|\substack{A^n=a^n\\ \Phi(a^n) = \bar{a}^n \\ F_n^{sys}(\bar{a}^n) = (\bar{a}^n,s^n)}\right\}\right]^\rho \\
	  &\overset{(d)}{=}  \sum\limits_{y^n} \left\{\sum\limits_{\bar{a}^n \in \typ{P_{\bar{A}}}} \frac{1}{\left|\typ{P_{\bar{A}}}\right|}P_A^n(a^n)^{-\frac{\rho}{1+\rho}} \sum\limits_{s^n} P_S^n(s^n) P(y^n|(\bar{a}^n,s^n))^{\frac{1}{1+\rho}}\right\} \times \\&\left\{\sum\limits_{\substack{\tilde{\tilde{a}}^n \in \mc{A}^n \\ \tilde{\tilde{a}}^n \neq \bar{a}^n}}  \sum\limits_{\bar{s}^n} P_S^n(\tilde{s}^n)  \left(P (y^n|(\tilde{\tilde{a}}^n,\tilde{s}^n)) P_A^n(\tilde{\tilde{a}}^n)\right)^{\frac{1}{1+\rho}}\right\}^\rho   \nonumber \\     
	  &\overset{(e)}{=}  \sum\limits_{y^n} \left\{\frac{2^{n\left(\kl{P_{\bar{A}}}{P_A} + \ent{P_{\bar{A}}}\right)}}{\left|\typ{P_{\bar{A}}} \right|}  \sum\limits_{\bar{a}^n \in \typ{P_{\bar{A}}}}\sum\limits_{s^n} P_S^n(s^n)\left(P_A^n(\bar{a}^n)P(y^n|(\bar{a}^n,s^n))\right)^{\frac{1}{1+\rho}}\right\} \times
	  \\& \left\{\sum\limits_{\substack{\tilde{\tilde{a}}^n \in \mc{A}^n \\ \tilde{\tilde{a}}^n \neq \bar{a}^n}} \sum\limits_{\tilde{s}^n}  P_S^n(\tilde{s}^n)\left(P_A^n(\tilde{\tilde{a}}^n) P (y^n|(\tilde{\tilde{a}}^n,\tilde{s}^n))\right)^{\frac{1}{1+\rho}}\right\}^\rho \nonumber \\ 
	  &\overset{(f)}{\leq} \frac{2^{n\left(\kl{P_{\bar{A}}}{P_A} + \ent{P_{\bar{A}}}\right)}}{\left|\typ{P_{\bar{A}}} \right|}  \sum\limits_{y^n} \left\{\sum\limits_{\bar{a}^n \in \mc{A}^n}\sum\limits_{s^n} P_S^n(s^n)\left(P_A^n(\bar{a}^n)P(y^n|(\bar{a}^n,s^n))\right)^{\frac{1}{1+\rho}}\right\} \times
	  \\& \left\{\sum\limits_{\tilde{\tilde{a}}^n \in \mc{A}^n} \sum\limits_{\tilde{s}^n}  P_S^n(\tilde{s}^n)\left(P_A^n(\tilde{\tilde{a}}^n) P (y^n|(\tilde{\tilde{a}}^n,\tilde{s}^n))\right)^{\frac{1}{1+\rho}}\right\}^\rho \nonumber \\ 
	  &= \frac{2^{n\left(\kl{P_{\bar{A}}}{P_A} + \ent{P_{\bar{A}}}\right)}}{\left|\typ{P_{\bar{A}}} \right|}  \sum\limits_{y^n} \left\{\sum\limits_{\bar{a}^n \in \mc{A}^n}\sum\limits_{s^n} P_S^n(s^n)\left(P_A^n(\bar{a}^n)P(y^n|(\bar{a}^n,s^n))\right)^{\frac{1}{1+\rho}}\right\}^{1+\rho} \\
	  &\overset{(g)}{=}\frac{2^{n\left(\kl{P_{\bar{A}}}{P_A} + \ent{P_{\bar{A}}}\right)}}{\left|\typ{P_{\bar{A}}} \right|}\left[\sum\limits_{y} \left\{\sum\limits_{a,s} P_S(s) \left\{P_A(a)P(y|(a,s)) \right\}^{\frac{1}{1+\rho}} \right\}^{1+\rho}\right]^n \\
	  &\overset{(h)}{=} 2^{-n\left(-\log\sum\limits_{y} \left\{\sum\limits_{a,s} P_S(s) \left\{P_A(a)P(y|(a,s)) \right\}^{\frac{1}{1+\rho}} \right\}^{1+\rho} -\alpha(n) -\kl{P_{\bar{A}}}{P_A}\right)}
	 \end{align}
	\begin{itemize}
		\item (a) uses Markov inequality and is true for any $\rho>0$ and $\eta >0$.
		\item (b) follows from the definitions of  $\Phi$ and $F_n^{sys}$.
		\item (c) follows from Jensen inequality and is true for $0 \leq \rho \leq 1$.
		\item (d) follows from rearrangement of the terms inside the expectation, the replacement $\eta = \frac{1}{1+\rho}$ and the fact that $P_A^n(a^n) = P_A^n(\phi(a^n))$ for all $a^n \in \mc{A}^n$.
		\item In (e) use the fact that $P_A^n(a^n) = P_A^n(\phi(a^n))$ for all $a^n \in \mc{A}^n$  and replace various probability values.
		\item In (f) we extend the summations to $\mc{A}^n$.
		\item (g) follows from product distributions and DMC.
		\item (h) follows by defining $\alpha(n) = \frac{|\mc{A}|\log(n+1)}{n}$ and using \eqref{eq:typsize}.
	\end{itemize}
\end{proof}
 Hence 
  \begin{align}
 P_{e,n} &\leq 2^{-n\left(-\log\sum\limits_{y} \left\{\sum\limits_{a,s} P_S(s) \left\{P_A(a)P(y|(a,s)) \right\}^{\frac{1}{1+\rho}} \right\}^{1+\rho} -\alpha(n) -\kl{P_{\bar{A}}}{P_A}\right)}  \label{eq:pasexp}
 \end{align}
 Optimizing over $\rho$ we get $E_{SM}$. Unfortunately $E_{SM}$ is not useful in establishing achievable rate bounds in the same way as we did for the setups in  Sec.~\ref{sec:sys} and Sec.~\ref{sec:mis}. This is because of the $\frac{\kl{P_{\bar{A}}}{P_A}}{\rho}$ term. 

 In Th. \ref{th:pasg} we introduced a permutation function $\phi_{P_{\bar{A}}}$ which is not the part of a standard PAS system. This function was introduced so that the random coding argument leads to meaningful error exponents; we do not claim that using such a permutation function would bring any gain in a practical system. Furthermore, for a uniform distribution over $\supp(Q_n)$ (as is the case usually for PAS since the input to the distribution matcher, i.e., the source message is uniformly distributed and the distribution matcher is a one-to-one mapping) and for symmetric channels one can show that the $\phi_{P_{\bar{A}}}$ is not needed.
\begin{corollary}
	In a standard PAS system we have $P_A = P_{\bar{A}}$. In this case $ \frac{\kl{P_{\bar{A}}}{P_A}}{\rho} =0$, hence we can use $E_{SM}$ to derive the achievable rates.  Using calculations analogous to the ones in Sec.~\ref{sec:sys} we conclude that $E_{SM} > 0$ as long as $\mi{AS}{Y} > \ent{P_{\bar{A}}}$. Similarly following the same arguments as in  Sec.~\ref{sec:mis} we have $R_n \leq \ent{P_{\bar{A}}}$. For \enquote{close} to uniform distribution over $\supp(Q_n)$ with $\frac{|\supp(Q_n)|}{|\typ{P_{\bar{A}}}|} \to 1$, any rate $R< \mi{X}{Y}$ is achievable, where $P_X=P_{\bar{A}}P_S$.
\end{corollary}	
\begin{remark}
Following the same lines as in Remark \ref{rem:opt}	we argue that $P_{\bar{A}}$ is a design choice. In the case when $P_A = P_{\bar{A}}$ this leads to the following maximum achievable rate expression. 
\begin{align}
R^* =  \max_{P_A \in \mc{\bar{P}_A},P_S} \mi{X}{Y}  
\end{align}
for $P_X=P_AP_S$  and $\mc{\bar{P}_A}$ being  the set of all distributions of finite type over $\mc{A}$ .
\end{remark}
\begin{remark}
	In practice the procedure to choose $P_{\bar{A}}$ is as follows: for a given channel $P_{Y|X}$, calculate the $P_{A^*}P_S$ which maximizes $\mi{X}{Y}$. Then for the chosen blocklength $n$ for communication, search for the "closest" approximation (for example in terms of divergence \cite{geigeroptimal})of $P_{A^*}$ among the set of all type-$n$ probability distributions. This approximation is then used as $P_{\bar{A}}$
\end{remark}
\begin{remark}
For the scenario in Corollary~\ref{rem:affine} we can restrict $\mc{H}_n$ to be the set of all affine mappings from $\mc{A}^n$ to $\mc{S}^n$ and define a uniform random ensemble $H_n$ of such mappings as done by Gallager \cite[Sec 6.2]{gallager1968information}.  This ensemble then can be used to prove Th.~\ref{th:pasg}.
\end{remark}
\section{Future Work}
 In the future we will analyze PAS for a more general class of distribution matchers. Future work will also focus on improving proof technique in Sec.~\ref{sec:paserror} such that one can also establish error exponent and achievable rates for $P_A \neq P_{\bar{A}}$.
\section*{Acknowledgment}
This work was supported by the German Federal Ministry of Education and Research in the framework of the Alexander von
Humboldt-Professorship. The author would like to thank Gerhard Kramer, Patrick Schulte and Georg B\"ocherer for their useful comments. 
\bibliographystyle{IEEEtran}
\bibliography{IEEEabrv,confs-jrnls,references}

\end{document}

%% file: preamble.tex
\usepackage{amsmath}
\usepackage{amsthm}
\usepackage{amssymb}
\usepackage{bm}
\usepackage{xspace}
\usepackage{xcolor}
\usepackage{graphicx}
\usepackage{url}
\usepackage{framed}
\usepackage{float}
\usepackage{rotating}
\usepackage{verbatim}
\usepackage{listings}
\usepackage{lscape}

\usepackage{pgfplots}
\usepackage{pgf}
\usepackage{tikz}
\usetikzlibrary{arrows,shapes.misc,chains,scopes}
\usetikzlibrary{matrix}
\usetikzlibrary{calc}
\usetikzlibrary{fit}

\usepackage{pgfplotstable}
\usepackage{booktabs}
\usepackage{colortbl}

\newcommand{\executeiffilenewer}[3]{%
\ifnum\pdfstrcmp{\pdffilemoddate{#1}}%
{\pdffilemoddate{#2}}>0%
{\immediate\write18{#3}}\fi%
}

\graphicspath{{figures/}}

\theoremstyle{plain}
\newtheorem{corollary}{Corollary}

\newtheorem{theorem}{Theorem}

\newcounter{algocount}
\newcounter{examplecount}

\newcommand{\bmm}{\begin{matrix}}
\newcommand{\emm}{\end{matrix}}
\newcommand{\bpm}{\begin{pmatrix}}
\newcommand{\epm}{\end{pmatrix}}

\newcommand{\bsbm}{\left[\begin{smallmatrix}}
\newcommand{\esbm}{\end{smallmatrix}\right]}
\newcommand{\bspm}{\left(\begin{smallmatrix}}
\newcommand{\espm}{\end{smallmatrix}\right)}

\newcommand{\bbm}{\begin{bmatrix}}
\newcommand{\ebm}{\end{bmatrix}}

\DeclareMathOperator*{\argmax}{argmax}

\DeclareMathOperator{\expop}{\mathbb{E}}

%% file: figures/classic/classic.tex
\tikzstyle{block} = [draw, rectangle, 
    minimum height=3em, minimum width=4em]
\tikzstyle{channel} = [draw, circle,  minimum width=3em]
\tikzstyle{input} = [coordinate]
\tikzstyle{output} = [coordinate]

\begin{tikzpicture}[auto, node distance=2cm,>=latex']
    \node [input, name=input] {};
    \node [block, right of=input] (encoder) {$f_n$};
    \node [channel, right of=encoder, node distance=3cm] (channel) {$P_{Y|X}$};
    \node [block, right of=channel, node distance=3cm] (decoder) {$g_n$};
    \node [output, right of=decoder] (out) {};
        
    \draw [->] (input) -- node {$W_n$} (encoder);
    \draw [->] (channel) -- node [name=y] {$Y^n$}(decoder);
    \draw [->] (encoder) -- node[name=u] {$X^n$} (channel);
    \draw [->] (decoder) -- node {$\hat{W}_n$} (out);
\end{tikzpicture}

%% file: figures/PASsetup/setup.tex
\tikzstyle{block} = [draw, rectangle, 
    minimum height=3em, minimum width=3em]
\tikzstyle{channel} = [draw, circle,  minimum width=3em]
\tikzstyle{input} = [coordinate]
\tikzstyle{output} = [coordinate]
\tikzstyle{mod} = [draw, circle, node distance=1cm]

\begin{tikzpicture}[auto, node distance=2cm,>=latex']
    \node [input, name=input] {};
    \node [block, right of=input, node distance=1.5cm] (matcher) {$\gamma$};
    \node [input, right of=matcher] (encoder1) {};
    \node [input, right of=encoder1, node distance=1cm] (imag) {};
    \node [block, right of=imag, node distance=1.3cm] (encoder) {$Mod$};
    \node [block, below of=imag, node distance = 1cm] (fec) {$h_n$};
    \node [channel, right of=encoder, node distance =4cm] (channel) {$P_{Y|X}$};
    \node [block, right of=channel, node distance=2cm] (decoder) {$g_n^{\scriptscriptstyle MM}$};
    \node[block, right of=decoder, node distance=2cm] (dematcher) {$\gamma^{-1}$};
     \node[output, right of=dematcher, node distance=1.5cm] (out) {};   
    \draw [->] (fec) -| node[near end, right]{$S^n$}  (encoder);
    \draw [->] (encoder1) |- (fec);
     \draw [->] (input) -- node {$B^\ell$} (matcher);
    \draw [->] (encoder1) -- (encoder);
    \draw [->] (matcher) -- node {$A^n \sim Q_n$} (encoder1);
    \draw [->] (channel) -- node [name=y] {$Y^n$}(decoder);
    \draw [->] (encoder) -- node[name=u] {$\scriptstyle X^n = (A_1S_1,\cdots, A_nS_n)$} (channel);
    \draw [->] (decoder) -- node {$\hat{A}^n$} (dematcher);
     \draw [->] (dematcher) -- node {$\hat{B}^\ell$} (out);
     
     \draw [color=blue,thick](3.45,-1.6) rectangle (6.4,0.6);
    \node at (4.5,0.5) [above=3mm, right=0mm] {$f_n^{sys}$};

\end{tikzpicture}

%% file: figures/PAStheory/setup.tex
\tikzstyle{block} = [draw, rectangle, 
    minimum height=3em, minimum width=3em]
\tikzstyle{channel} = [draw, circle,  minimum width=3em]
\tikzstyle{input} = [coordinate]
\tikzstyle{output} = [coordinate]
\tikzstyle{mod} = [draw, circle, node distance=1cm]

\begin{tikzpicture}[auto, node distance=2cm,>=latex']
    \node [input, name=input] {};
    \node [input, right of=input, node distance=1.5cm] (encoder1) {};
    \node [input, right of=encoder1, node distance=1cm] (imag) {};
    \node [block, right of=imag, node distance=1.3cm] (encoder) {$Mod$};
    \node [block, below of=imag, node distance = 1cm] (fec) {$h_n$};
    \node [channel, right of=encoder, node distance =4cm] (channel) {$P_{Y|X}$};
    \node [block, right of=channel, node distance=2cm] (decoder) {$g_n^{\scriptscriptstyle MMAP}$};
    \node [output, right of=decoder, node distance=1.3cm] (out) {};
        
    \draw [->] (fec) -| node[near end, right]{$S^n$}  (encoder);
    \draw [->] (encoder1) |- (fec);
    \draw [->] (encoder1) -- (encoder);
    \draw [->] (input) -- node {$A^n \sim Q_n$} (encoder1);
    \draw [->] (channel) -- node [name=y] {$Y^n$}(decoder);
    \draw [->] (encoder) -- node[name=u] {$\scriptstyle X^n =(A_1S_1,\cdots, A_nS_n)$} (channel);
    \draw [->] (decoder) -- node {$\hat{A}^n$} (out);
    
     \draw [color=blue,thick](1.4,-1.6) rectangle (4.4,0.6);
    \node at (3.25,0.5) [above=3mm, right=0mm] {$f_n^{sys}$};

\end{tikzpicture}

%% file: figures/sys/setup.tex
\tikzstyle{block} = [draw, rectangle, 
    minimum height=3em, minimum width=3em]
\tikzstyle{channel} = [draw, circle,  minimum width=3em]
\tikzstyle{input} = [coordinate]
\tikzstyle{output} = [coordinate]
\tikzstyle{mod} = [draw, circle, node distance=1cm]

\begin{tikzpicture}[auto, node distance=2cm,>=latex']
    \node [input, name=input] {};
    \node [input, right of=input, node distance=1.5cm] (encoder1) {};
    \node [input, right of=encoder1, node distance=1cm] (imag) {};
    \node [block, right of=imag, node distance=1.3cm] (encoder) {$Mod$};
    \node [block, below of=imag, node distance = 1cm] (fec) {$h_n$};
    \node [channel, right of=encoder, node distance =4cm] (channel) {$P_{Y|X}$};
    \node [block, right of=channel, node distance=2cm] (decoder) {$g_n^{\scriptscriptstyle MAP}$};
    \node [output, right of=decoder, node distance=1.4cm] (out) {};
        
    \draw [->] (fec) -| node[near end, right]{$S^n$}  (encoder);
    \draw [->] (encoder1) |- (fec);
    \draw [->] (encoder1) -- (encoder);
    \draw [->] (input) -- node {$A^n\sim P_A^n$} (encoder1);
    \draw [->] (channel) -- node [name=y] {$Y^n$}(decoder);
    \draw [->] (encoder) -- node[name=u] {$\scriptstyle X^n =(A_1S_1,\cdots, A_nS_n)$} (channel);
    \draw [->] (decoder) -- node {$\hat{A}^n$} (out);
    
     \draw [color=blue,thick](1.45,-1.6) rectangle (4.4,0.6);
    \node at (3.25,0.5) [above=3mm, right=0mm] {$f_n^{sys}$};

\end{tikzpicture}

%% file: figures/mismatched/mis.tex
\tikzstyle{block} = [draw, rectangle, 
    minimum height=3em, minimum width=4em]
\tikzstyle{channel} = [draw, circle,  minimum width=3em]
\tikzstyle{input} = [coordinate]
\tikzstyle{output} = [coordinate]

\begin{tikzpicture}[auto, node distance=2cm,>=latex']
    \node [input, name=input] {};
    \node [block, right of=input] (encoder) {$f_n$};
    \node [channel, right of=encoder, node distance=2.5cm] (channel) {$P_{Y|X}$};
    \node [block, right of=channel, node distance=2.5cm] (decoder) {$g_n^{\scriptscriptstyle MMAP}$};
    \node [output, right of=decoder, node distance=1.5cm] (out) {};
        
    \draw [->] (input) -- node {$A^n \sim Q_n$} (encoder);
    \draw [->] (channel) -- node [name=y] {$Y^n$}(decoder);
    \draw [->] (encoder) -- node[name=u] {$X^n$} (channel);
    \draw [->] (decoder) -- node {$\hat{A}^n$} (out);
\end{tikzpicture}